\documentclass[11pt]{article}
\usepackage{fullpage}
\usepackage{here}

\usepackage{amsthm,amsmath,amssymb}
\usepackage{xcolor}
\usepackage{graphicx}
\usepackage{latexsym}
\usepackage{xspace}
\usepackage{todonotes}
\usepackage{caption}
\usepackage{subcaption}
\usepackage{enumitem}
\usepackage[vlined,linesnumbered]{algorithm2e}
\usepackage{hyperref}
\usepackage{lineno}

\usepackage{thm-restate}

\theoremstyle{plain}

\newtheorem{theorem}{Theorem}
\newtheorem{lemma}[theorem]{Lemma}
\newtheorem{proposition}[theorem]{Proposition}
\newtheorem{corollary}[theorem]{Corollary}

\newcommand{\ALG}{\textsc{Recognize}}

\newcommand{\wdth}{{\rm width}}
\newcommand{\hght}{{\rm height}}
\newcommand{\LL}{\mathcal{L}}
\newcommand{\layout}{$\mathcal{L}$\xspace}
\newcommand{\layoutp}{$\mathcal{L'}$\xspace}
\newcommand{\layoutstar}{$\mathcal{L^*}$\xspace}

  {\begin{list}{}%
          {\setlength{\leftmargin}{#1}}%
          \item[]%
  }
  {\end{list}}

\hypersetup{
hidelinks,
colorlinks=true,
citecolor=[rgb]{0.121 0.47 0.705},
linkcolor=[rgb]{0.7 0 0.4},
urlcolor=[rgb]{0.121 0.47 0.705}
}

\begin{document}

\title{Aspect Ratio Universal Rectangular Layouts\thanks{A preliminary version of this paper appeared in the \emph{Proceedings of the 16th International Conference and Workshops on Algorithms and Computation (WALCOM~2022)}, LNCS~13174, pp.~73--84, \href{http://dx.doi.org/10.1007/978-3-030-96731-4_7}{doi:10.1007/978-3-030-96731-4\_7}. Research on this paper was partially supported by the NSF award DMS-1800734.}
}

\author{Stefan Felsner\thanks{Institut f\"ur Mathematik, Technische Universit\"at Berlin, Berlin, Germany. Email: \texttt{felsner@math.tu-berlin.de}}
\and Andrew Nathenson\thanks{University of California, San Diego, CA, USA. Email: \texttt{anathenson@ucsd.edu}}
\and Csaba D. T\'oth\thanks{Department of Mathematics, California State University Northridge, Los Angeles, CA; and Department of Computer Science, Tufts University, Medford, MA, USA. Email: \texttt{csaba.toth@csun.edu}} %\orcidID{0000-0002-8769-3190
}

\date{}

%\linenumbers

\maketitle

\begin{abstract}
A \emph{generic rectangular layout} (for short, \emph{layout}) is a subdivision of an axis-aligned rectangle into axis-aligned rectangles, no four of which have a point in common. Such layouts are used in data visualization and in cartography. The contacts between the rectangles represent semantic or geographic relations.
A layout is weakly (strongly) \emph{aspect ratio universal} if any assignment of aspect ratios to rectangles can be realized by a weakly (strongly) equivalent layout. We give combinatorial characterizations for weakly and strongly aspect ratio universal layouts. Furthermore, we describe a quadratic-time algorithm that decides whether a given graph is the dual graph of a strongly aspect ratio universal layout, and finds such a layout if one exists.
%
%Keywords: rectangular layouts, contact graphs, universality.
\end{abstract}

\section{Introduction}
\label{sec:intro}

A \emph{rectangular layout} (for short, \emph{layout}) is a subdivision of an axis-aligned rectangle into axis-aligned rectangular faces; it is also known as \emph{mosaic floorplan} or \emph{rectangulation}. A layout is \emph{generic} if no four faces have a point in common; see Fig.~\ref{fig:intro1}. In this paper all layouts are generic unless stated otherwise.
In the \emph{dual graph} $G(\mathcal{L})$ of a layout $\mathcal{L}$, the nodes correspond to rectangular faces, and an edge corresponds to a pair of faces that share a boundary segment of positive length~\cite{EppsteinMSV12,Rinsma87a,Rinsma87b}.

\begin{figure}[htbp]
	\centering
    \includegraphics[width=0.9\textwidth]{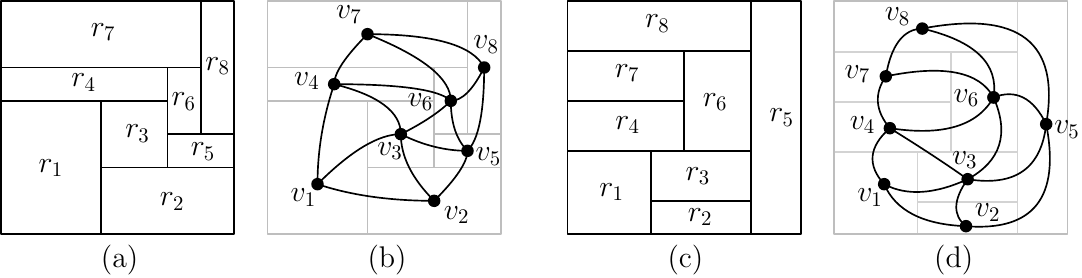}
	\caption{(a--b) A layout and its dual graph. (c--d) A sliceable layout and its dual graph.
	The two layouts are neither strongly nor weakly equivalent,
	but their dual graphs are isomorphic.} \label{fig:intro1}
\end{figure}

Two layouts are \emph{strongly equivalent} if they have isomorphic dual graphs, and the corresponding line segments between adjacent faces have the same orientation (horizontal or vertical).
Two layouts are \emph{weakly equivalent} if there are bijections between their maximal horizontal line segments as well as between their maximal vertical line segments, and the contact graphs of all segments are isomorphic plane graphs.
Strong equivalence implies weak equivalence~\cite{Felsner14}, but weak equivalence does not imply strong equivalence; see Fig.~\ref{fig:intro2} for examples. The closures of weak (resp., strong) equivalence classes under the Hausdorff distance\footnote{The distance between layouts $\LL_1$ and $\LL_2$ is the Hausdorff distance between the sets $S(\LL_1)$ and $S(\LL_2)$, where $S(\LL_i)$ is the union of all segments in $\LL_i$ for $i\in \{1,2\}$.} extend to nongeneric layouts, and a nongeneric layout may belong to the closures of multiple equivalence classes.

\begin{figure}[htbp]
	\centering
    \includegraphics[width=0.9\textwidth]{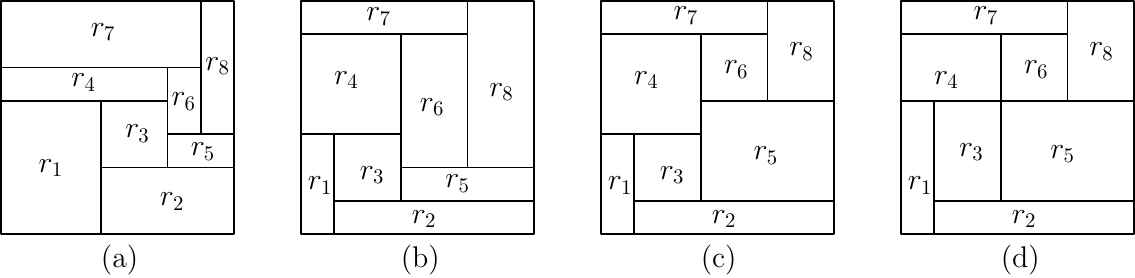}
	\caption{(a--c) Three generic layouts that are weakly equivalent. The layouts in (a) and (b) are strongly equivalent, but not strongly equivalent to the layout in (c), as the adjacencies between the pairs of rectangles $\{r_3,r_6\}$ and $\{r_4,r_5\}$ are different. (d) A nongeneric layout that is in the closure of the strong equivalence classes of the layouts in (a--b) and (c). \label{fig:intro2}}
\end{figure}

Rectangular layouts have been studied for more than 40 years, originally motivated by VLSI design~\cite{MSL76,Otten82,WimerKC88} and cartography~\cite{Raisz34}, and more recently by data visualization~\cite{KreveldS07}.
%They are of interest from the perspective of geometric intersection graphs~\cite{Felsner2013,Lovasz,Schramm-Squares}.
The weak equivalence classes of layouts are in bijection with Baxter permutations~\cite{AckermanBP06,Reading12,YaoCCG03}.

A graph is called a \emph{proper graph} if it is the dual graph of a generic layout. Every proper graph is a near-triangulation (a plane graph where every bounded face is a triangle, but the outer face need not be a triangle). But not every near-triangulation is a proper graph~\cite{Rinsma87a,Rinsma87b}. Ungar~\cite{Ungar53} gave a combinatorial characterization of proper graphs (see also~\cite{KozminskiK85,Thomassen84}): A near-triangulation is a proper graph if and only if it can be augmented by four vertices into a near-triangulation in which the outer face is a quadrilateral formed by the new vertices, and there are no separating triangles.
Using this characterization, proper graphs can be recognized in linear time~\cite{Hasan0K13,Nishizeki013,RahmanNN98,RahmanNN02}.

In data visualization and cartography~\cite{KreveldS07,Raisz34}, the rectangular faces correspond to entities (e.g., countries or geographic regions); adjacency between rectangles represents semantic or geographic relations, and the ``shape'' (e.g., area, width, height, aspect ratio, in-radius, etc.) of a rectangular face represents data associated with the entity. It is often desirable to use  equivalent layouts to realize different statistics associated with the same entities.
A generic layout \layout is \emph{weakly} (\emph{strongly}) \emph{area universal} if any area assignment to the rectangles can be realized by a layout weakly (strongly) equivalent to \layout. Wimer et al.~\cite{WimerKC88} showed that every generic layout is weakly area universal (see also~\cite[Thm.~3]{Felsner14}).
Eppstein et al.~\cite{EppsteinMSV12} proved that a layout is strongly area universal if and only if it is one-sided (defined below). However, no polynomial-time algorithm is known for testing whether a given graph $G$ is the dual graph of some area-universal layout.

\paragraph{Aspect ratio universal layouts.}
The \emph{aspect ratio} of an axis-aligned rectangle $r=[a,b]\times [c,d]$ is $\hght(r)/\wdth(r)=(d-c)/(b-a)$.
In some applications, the aspect ratios (rather than the areas) of the rectangles are specified. For example, in word clouds adapted to multiple languages, the aspect ratio of (the bounding box of) each word depends on the particular language.
A generic layout $\mathcal{L}$ is \emph{weakly} (\emph{strongly}) \emph{aspect ratio universal} (\emph{ARU}, for short) if any assignment of aspect ratios to the rectangles can be realized by a layout that is weakly (strongly) equivalent to \layout.

In this paper, we give a combinatorial characterization for weakly and strongly ARU layouts, and design an algorithm to recognize strongly ARU layouts. Our results are stated in Section~\ref{ssec:results}.

\subsection{Background and Terminology}
\label{ssec:def}
Before we can state our results, we review some additional key  definitions for layouts.
A \emph{rectilinear graph} is a plane straight-line graph in which every edge is horizontal or vertical.
A \emph{rectangular layout} (for short, \emph{layout}) is a rectilinear graph in which the boundary of every face (including the outer face) is a rectangle; it is \emph{generic} if the maximum vertex degree is at most three. A \emph{sublayout} of a layout \layout is a subgraph of~\layout which is a layout. A sublayout $\LL'$ of $\LL$ is \emph{trivial} if $\LL'=\LL$ or $\LL'$ is the boundary of a single face of $\LL$.
A layout $\LL$ is \emph{irreducible} if it does not have any nontrivial sublayout.
A \emph{rectangular arrangement} is a 2-connected plane straight-line graph in which every bounded face is a rectangle (however, the outer face need not be a rectangle); see Fig.~\ref{fig:intro3} for examples.

\begin{figure}[htbp]
	\centering
    \includegraphics[width=0.9\textwidth]{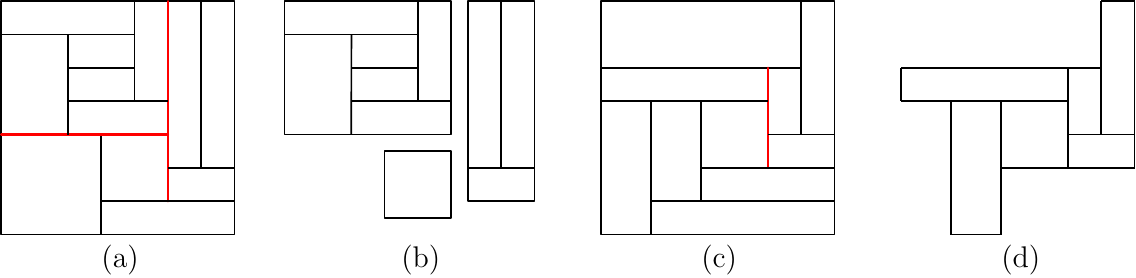}
	\caption{(a) A layout $\LL$, the red maximal segments are not one-sided. (b) Three sublayouts of $\LL$, one of which is trivial.
	(c) An irreducible layout.
        (d) A rectangular arrangement.\label{fig:intro3}}
\end{figure}

\paragraph{One-Sided Layouts.}
A \emph{segment} of a layout \layout is a path of collinear inner edges of \layout. A segment of \layout that is not contained in any other segment is \emph{maximal}. A layout \layout is \emph{one-sided} if every maximal segment is a side of at least one rectangular face of \layout; that is, for every maximal segment $s$, all other segments with an endpoint in the interior of $s$ lie in the same halfplane bounded by $s$.
For example, the layouts in Fig.~\ref{fig:intro3}a and~\ref{fig:intro3}c
are not one-sided (due to the red maximal segments).

\paragraph{Sliceable Layouts.}
A \emph{slice} in a layout \layout is a maximal segment with both endpoints incident to the outer face, and subdivides \layout into two sublayouts. A \emph{sliceable layout} (also known as \emph{slicing floorplan} or \emph{guillotine rectangulation}) is one that can be  decomposed into trivial layouts through recursive subdivision with slices; see Fig~\ref{fig:intro1}(c) for an example. The recursive subdivision of a layout $\LL$ can be represented by a \emph{binary space partition tree} (\emph{BSP-tree}), which is a binary tree where each vertex is associated with a sublayout, which is $\LL$ if the vertex is the root and a rectangular face of $\LL$ if the vertex is a leaf~\cite{deBerg2008}. For a nonleaf vertex, the tree additionally stores a slice; and two sublayouts on each side of the slice are associated with the two children. The \emph{root slice} is between two opposite sides of the bounding box.

The number of (strong equivalence classes of) sliceable layouts with $n$ rectangular faces is known to be the $n$th Schr\"oder number~\cite{YaoCCG03}. One-sided sliceable layouts are in bijection with certain pattern-avoiding permutations, closed formulas for their number have been given by Asinowski and Mansour~\cite{AsinowskiM10}; see also~\cite{MerinoM23} and \href{https://oeis.org/A078482}{OEIS~A078482} in the on-line encyclopedia of integer sequences for further references.

A \emph{windmill} in a layout is a set of four pairwise noncrossing maximal segments, called \emph{arms}, which contain the sides of a central rectangle and each arm has an endpoint on the interior of another (e.g., the maximal segments around the rectangular face $r_3$ or $r_6$ in Fig.~\ref{fig:transversal}a). A windmill is either \emph{clockwise} or \emph{counterclockwise}, depending on the orientation of the boundary of the central rectangle
when we orient each arm from the central rectangle to the other endpoint (see Figs.~\ref{fig:bad-sliceable}e--\ref{fig:bad-sliceable}f). It is well known that a layout is sliceable if and only if it does not contain a windmill~\cite{AckermanBP06}.

\paragraph{Transversal Structure.}
The dual graph $G(\LL)$ of a layout \layout encodes adjacencies between faces, but does not specify the relative positions between faces (above-below or left-right). The transversal structure (also known as \emph{regular edge-labelling}) was introduced by He~\cite{He93,KantH97} for the efficient recognition of proper graphs, and later used extensively for counting and enumerating (equivalence classes of) layouts. Fusy~\cite[Section~2]{Fusy09} distinguishes between several equivalent variants of the transversal structure, we refer to the variant called \emph{transversal edge-partition} by Fusy~\cite{Fusy09}.

We first define the  \emph{extended dual graph} $G^*(\mathcal{L})$ of a layout $\LL$, as the contact graph of the rectangular faces \emph{and} the four edges of the bounding box of $\mathcal{L}$; it is a triangulation in an outer 4-cycle without separating triangles; see Fig.~\ref{fig:transversal}.

\begin{figure}[htbp]
	\centering
    \includegraphics[width=0.8\textwidth]{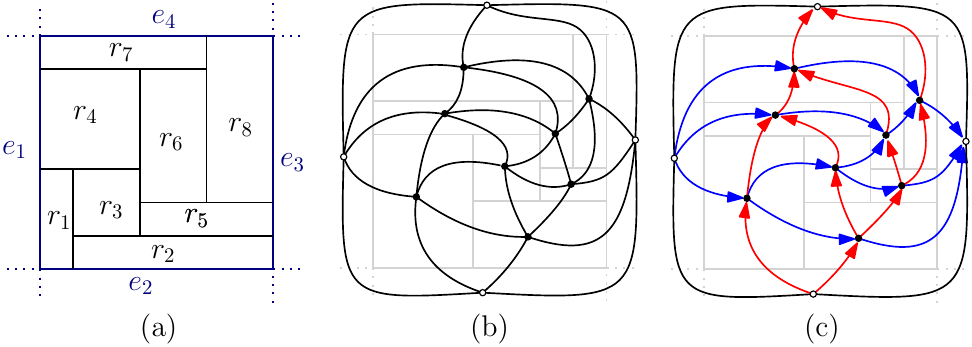}
	\caption{(a) A layout \layout bounded by edges $e_1,\ldots , e_4$.
	(b) Extended dual graph $G^*(\LL)$ with an outer 4-cycle $(e_1,\ldots ,e_4)$.
	(c) A transversal structure.}\label{fig:transversal}
\end{figure}

The \emph{transversal structure} of a layout $\LL$ comprises $G^*(\mathcal{L})$ and a bicoloring of the inner edges of $G^*(\mathcal{L})$, where red (resp., blue) edges correspond to above-below (resp., left-right) relation between two objects in contact.
An (abstract) \emph{transversal structure} is defined as a graph $G^*$, which is a triangulation of an outer 4-cycle $(S,W,N,E)$ that has no separating triangles, together with a bicoloring of the inner edges of $G^*$ such that all the inner edges incident to $S$, $W$, $N$, and $E$ are red, blue, red, and blue, respectively; and at each inner vertex the counterclockwise rotation of incident edges consists of four nonempty blocks of red, blue, red, and blue edges; see Fig.~\ref{fig:transversal}c.
The transversal structure determines a unique orientation of the edges~\cite[Proposition~2]{Fusy09}, where red (resp., blue) edges are directed bottom-up (resp., left-to-right), as indicated in Fig.~\ref{fig:transversal}c.
It is known that transversal structures are in bijection with the strong equivalence classes of generic layouts~\cite{Felsner2013,Fusy09,KantH97}.

\begin{figure}[htbp]
	\centering
    \includegraphics[width=0.8\textwidth]{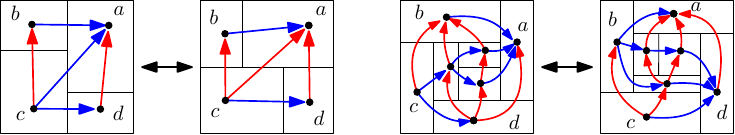}
	\caption{A flip of an empty (left) and a nonempty (right) alternating cycle.}\label{fig:flips}
\end{figure}

\paragraph{Flips and Alternating 4-Cycles.}
A sequence of flip operations can transform any transversal structure with $n$ inner vertices into any other~\cite{Felsner04,Fusy09}. Each \emph{flip} considers an \emph{alternating 4-cycle} $C$, which comprises red and blue edges alternatingly, and changes the color of every edge in the interior of $C$; see Fig.~\ref{fig:flips}. If, in particular, there is no vertex in the interior of $C$, then the flip changes the color of the inner diagonal of $C$.
Furthermore, every flip operation yields a valid transversal structure on $G^*(\LL)$, hence a new generic layout $\LL'$ that is strongly nonequivalent to \layout.
We can now establish a relation between geometric and combinatorial properties.

\begin{lemma}\label{lem:flip}
A layout \layout is one-sided and sliceable if and only if $G^*(\LL)$ admits a unique transversal structure.
\end{lemma}
\begin{proof}
Assume that \layout is a layout where $G^*(\LL)$ admits two or more transversal structures. Consider a transversal structure of $G^*(\LL)$. Since any two transversal structures are connected by a sequence of flips, there exists an alternating 4-cycle. Any alternating 4-cycle with no interior vertex corresponds to a segment in \layout that is not one-sided. Any alternating 4-cycle with interior vertices corresponds to a windmill in \layout.
Consequently, \layout is not one-sided or not sliceable.

Conversely, if \layout is not one-sided (resp., sliceable), then the transversal structure of $G^*(\LL)$ contains an alternating 4-cycle with no interior vertex (resp., with interior vertices). Consequently, we can perform a flip operation, and obtain another transversal structure for $G^*(\LL)$.
\end{proof}

\subsection{Our Results}
\label{ssec:results}
We characterize strongly and weakly aspect ratio universal layouts.

\begin{restatable}{theorem}{weaktheorem}
%\begin{theorem}
\label{thm:weak}
A generic layout is weakly aspect ratio universal if and only if it is sliceable.
%\end{theorem}
\end{restatable}

\begin{restatable}{theorem}{strongtheorem}
%\begin{theorem}
\label{thm:equivalence}
For a generic layout  \layout, the following properties are equivalent:
\begin{enumerate}
    \item[(i)]  \layout is strongly aspect ratio universal;
    \item[(ii)]  \layout is one-sided and sliceable;
    \item[(iii)] the extended dual graph $G^*(\LL)$ of \layout admits a unique transversal structure.
\end{enumerate}
%\end{theorem}
\end{restatable}
It is not difficult to show that one-sided sliceable layouts are strongly aspect ratio universal; and admit a unique transversal structure. Proving the converses, however, is more involved.

\paragraph{Algorithmic Results.}
In some applications, the rectangular layout is not specified, and we are only given the dual graph of a layout (i.e., a proper graph). This raises the following problem: Given a (proper) graph $G$, find a strongly (resp., weakly) ARU layout $\mathcal{L}$ such that $G$ is isomorphic to the dual graph of $\LL$ (i.e., $G\simeq G(\mathcal{L})$) or report that no such layout exists. Using structural properties of one-sided sliceable layouts that we develop here, we present an algorithm for recognizing the dual graphs of strongly ARU layouts.

\begin{restatable}{theorem}{algotheorem}
%\begin{theorem}
\label{thm:algorithm}
We can decide in $O(n^2)$ time whether a given graph $G$ with $n$ vertices is the dual of a one-sided sliceable layout.
%\end{theorem}
\end{restatable}

Currently, no polynomial-time algorithm is known for recognizing dual graphs of sliceable layouts~\cite{DasguptaS01,KustersS15,YeapS95} (which are weakly ARU layouts by Theorem~\ref{thm:weak}); or one-sided layouts~\cite{EppsteinMSV12}.
Previously, Thomassen~\cite{Thomassen84} gave a linear-time algorithm to recognize proper graphs if the nodes corresponding to corner rectangles are specified, using combinatorial characterizations of layouts~\cite{Ungar53}. Kant and He~\cite{He93,KantH97} described a linear-time algorithm to test whether a given graph $G^*$ is the extended dual graph of a layout, using transversal structures. Later, Rahman et al.~\cite{Hasan0K13,Nishizeki013,RahmanNN98,RahmanNN02} showed that proper graphs can be recognized in linear time (without specifying the corners). However, a proper graph may have exponentially many (strongly or weakly) nonequivalent realizations, and none of the prior algorithms guarantees to find a one-sided sliceable realization if one exists.

\paragraph{Organization.} We characterize strongly and weakly ARU layouts and prove Theorems~\ref{thm:weak}--\ref{thm:equivalence} in Section~\ref{sec:characterization}. We establish structural properties of the dual graphs of one-sided sliceable layouts in Subsection~\ref{ssec:structure}, and use them in the analysis of a quadratic-time algorithm that recognizes such graphs in Subsection~\ref{ssec:alg}. We conclude in Section~\ref{sec:con}  with open problems and one more structural property: We show that the dual graph of every one-sided sliceable layout has a vertex cut of size at most 3 (Proposition~\ref{pro:cuts}).

\section{Aspect Ratio Universality}
\label{sec:characterization}

An \emph{aspect ratio assignment} to a layout $\mathcal{L}$ is a function $\alpha$ that maps a positive real to each rectanglular face in $\mathcal{L}$. A strong (resp., weak) \emph{realization} of an aspect ratio assignment $\alpha$ to $\mathcal{L}$ is a layout $\mathcal{L}'$ that is strongly (resp., weakly) equivalent to $\LL$ with the required aspect ratios. If such a realization exists, then the aspect ratio assignment $\alpha$ is strongly (resp., weakly) \emph{realizable}. A layout is strongly (resp., weakly) \emph{aspect ratio universal} (\emph{ARU}) if every aspect ratio assignment is strongly (resp., weakly) realizable. In this section, we characterize strongly and weakly ARU layouts (Theorems~\ref{thm:weak} and~\ref{thm:equivalence}). We start with two easy observations about one-sided and sliceable layouts.

\begin{lemma}\label{lem:strong-weak}
Let \layout be a layout with an aspect ratio assignment $\alpha$.
Then every strong realization is a weak realization.
Furthermore, if \layout is one-sided, then every weak realization is a strong realization.
\end{lemma}
\begin{proof}
The first statement follows from the fact that strong equivalence implies weak equivalence.
For the second statement, note that in a one-sided layout, each maximal segment $s$ is a side of a rectangular face, and so $s$ is on the boundary between the same pairs of rectangular faces in every weak realization of $\alpha$. Consequently, all weak realizations of $\alpha$ generate the same dual graph, hence they are all strongly equivalent to \layout, and so they are strong realizations of $\alpha$.
\end{proof}

\begin{lemma}
\label{lem:sliceable-uniqueness}
Every aspect ratio assignment for a sliceable layout has a unique weak realization up to scaling and translation. Furthermore, for every $\varrho>0$ there exists a weakly realizable aspect ratio assignment $\alpha$ such that the bounding box of the unique realization of $\alpha$ has aspect ratio $\varrho$.
\end{lemma}
\begin{proof}
To prove the first claim, let $\alpha$ be an aspect ratio assignment to a sliceable layout \layout. Note that the restriction of $\alpha$ to a sublayout $\LL'$ of \layout is an aspect ratio assignment for $\LL'$. We proceed by induction on $k$, the height of the BSP-tree representing \layout.

In the basis step, a layout of height 0 comprises a layout with a single rectangular face, which is uniquely determined by its aspect ratio assignment up to scaling and translation. For the induction step, assume that every sublayout at height $k$ in the BSP-tree admits a unique weak realization in which all rectangular faces (at the leaves of the BSP-tree) have the required aspect ratios. A sublayout $\LL_0$ at height $k+1$ of the BSP-tree is composed of two sublayouts at height $k$, say $\LL_1$ and $\LL_2$, that share an edge. Given a weak realization of $\LL_1$, there is a unique scaling and translation that attaches $\LL_2$ to $\LL_1$, and identifies their matching boundary segments. This yields a unique weak realization of the sublayout $\LL$ at level $k+1$, up to scaling and translation.

The second claim follows trivially: Start with a bounding box of aspect ratio $\varrho$, subdivide it recursively into a layout equivalent to \layout (but otherwise arbitrarily), and define an aspect ratio assignment using the aspect ratios of the resulting leaf rectangles.
\end{proof}

By the definition of weak ARU, Lemma~\ref{lem:sliceable-uniqueness} readily implies the following.

\begin{corollary}
\label{cor:x12}
If \layout is sliceable, then it is weakly ARU.
\end{corollary}

The combination of Lemmata~\ref{lem:strong-weak} and~\ref{lem:sliceable-uniqueness} yields the following for one-sided sliceable layouts.

\begin{corollary}
\label{cor:111}
If \layout is one-sided and sliceable, then it is strongly ARU.
\end{corollary}

\subsection{Sliceable and One-Sided Layouts}
\label{ssec:first}

Next we show that any sliceable layout that is strongly ARU must be one-sided.
We present two types of elementary layouts that are not strongly ARU, and then show that all other layouts that are not one-sided or not sliceable can be reduced to these prototypes.
A \emph{brick layout} is a sliceable layout whose BSP-tree is a complete binary tree with four leaves, and the slices at consecutive levels are orthogonal; see Figs.~\ref{fig:brick1}--\ref{fig:brick4}. A \emph{windmill layout} is generated by the four arms of a (clockwise or counterclockwise) windmill, where each arm extends to the outer boundary; see Figs.~\ref{fig:windmill1}--\ref{fig:windmill2}.

\begin{figure}[htbp]	
	\centering
	\begin{subfigure}[t]{0.12\textwidth}
		\centering
		\includegraphics[scale=0.75]{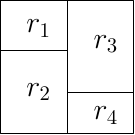}
		\caption{\label{fig:brick1}}
	\end{subfigure}
	\quad
    \begin{subfigure}[t]{0.12\textwidth}
		\centering
		\includegraphics[scale=0.75]{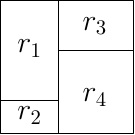}
		\caption{\label{fig:brick2}}
	\end{subfigure}
	\quad
	\begin{subfigure}[t]{0.12\textwidth}
		\centering
		\includegraphics[scale=0.75]{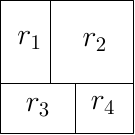}
		\caption{\label{fig:brick3}}
	\end{subfigure}
	\quad
	\begin{subfigure}[t]{0.12\textwidth}
		\centering
		\includegraphics[scale=0.75]{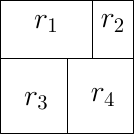}
		\caption{\label{fig:brick4}}
	\end{subfigure}
	\quad
		\begin{subfigure}[t]{0.12\textwidth}
		\centering
		\includegraphics[scale=0.75]{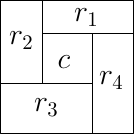}
		\caption{\label{fig:windmill1}}
	\end{subfigure}
	\quad
		\begin{subfigure}[t]{0.15\textwidth}
		\centering
		\includegraphics[scale=0.75]{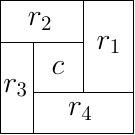}
		\caption{\label{fig:windmill2}}
	\end{subfigure}
	\caption{Layouts that are not aspect ratio universal: (a)--(d) brick layouts are sliceable but not one-sided; (e)--(f) windmill layouts are one-sided but not sliceable. \label{fig:bad-sliceable}}
\end{figure}

\begin{lemma}\label{lem:prototypes}
The brick layouts in Figs.~\ref{fig:brick1}--\ref{fig:brick4} are not strongly ARU; the windmill layouts in Figs.~\ref{fig:windmill1}--\ref{fig:windmill2} are neither strongly nor weakly ARU.
\end{lemma}
\begin{proof}
Suppose, for a contradiction, that a brick layout is strongly ARU. We may assume, by symmetry, that the brick layout $\mathcal{L}_0$ in Fig.~\ref{fig:brick1} is strongly ARU. Then there exists a strongly equivalent layout $\mathcal{L}$ for the aspect ratio assignment $\alpha(r_2)=\alpha(r_3)=1$ and $\alpha(r_1)=\alpha(r_4)=2$.
Since the same horizontal slice is a side of both $r_1$ and $r_2$,
then $\wdth(r_1)=\wdth(r_2)$. Combined with $\alpha(r_1)=2\alpha(r_2)$, this yields $\hght(r_1)=2\,\hght(r_2)$, and so the left horizontal slice is below the barycenter of $r_1\cup r_2$. Similarly, $\wdth(r_3)=\wdth(r_4)$ and $\alpha(r_4)=2\alpha(r_3)$ imply that the right horizontal slice is above the barycenter of $r_3\cup r_4$. Since barycenter of $r_1\cup r_2$ and $r_3\cup r_4$ have the same $y$-coordinate, then $r_1$ and $r_4$ are in contact in the realization $\LL$, which is not strongly equivalent to $\mathcal{L}_0$: a contradiction.

Suppose that a windmill layout is weakly ARU. We may assume, by symmetry, that the windmill layout $\mathcal{L}_1$ in Fig.~\ref{fig:windmill1} is weakly ARU. Then there exists a weakly equivalent layout $\mathcal{L}$ for the aspect ratio assignment $\alpha(c)=\alpha(r_1)=\alpha(r_2)=\alpha(r_3)=\alpha(r_4)=1$.
In particular, $r_1,\ldots ,r_4$ are squares; denote their side lengths by $s_i$, for $i=1,\ldots ,4$. Note that one side of $r_i$ strictly contains a side of $r_{i-1}$ for $i=1,\ldots , 4$ (with arithmetic modulo 4). Consequently, $s_1<s_2<s_3<s_4<s_1$, which is a contradiction.
\end{proof}

\begin{lemma}\label{lem:222}
If a layout is sliceable but not one-sided, then it is not strongly ARU.
\end{lemma}
\begin{proof}
Let $\mathcal{L}$ be a sliceable but not one-sided layout. It suffices to show that any of its sublayouts are not strongly ARU, because any nonrealizable aspect ratio assignment for a sublayout can be expanded arbitrarily to an aspect ratio assignment for the entire layout.

We claim that $\mathcal{L}$ contains a sublayout strongly equivalent to a brick layout in Figs.~\ref{fig:brick1}--\ref{fig:brick4}. As \layout is not one-sided, it contains a maximal segment $\ell$ which is not the side of any rectangular face. We may assume, without loss of generality, that $\ell$ is vertical. Since \layout is sliceable, every maximal segment is a slice of a sublayout that subdivides it into two smaller sublayouts.
However, $\ell$ is not the side of any rectangular face, consequently the two smaller sublayouts on the left and right of $\ell$ must be subdivided horizontally in the recursion. Let $\ell_{\rm left}$ and $\ell_{\rm right}$ be the first maximal horizontal segments on the left and right of $\ell$, respectively. Assume that they each subdivide a sublayout adjacent to $\ell$ into $r_1$ and $r_2$ (on the left) and $r_3$ and $r_4$ (on the right). The bounding boxes of the sublayouts $r_1,\ldots , r_4$ comprise a layout equivalent to a brick layout in Figs.~\ref{fig:brick1}--\ref{fig:brick4}.
By Lemma~\ref{lem:prototypes}, there exists an aspect ratio assignment to $\mathcal{L}$ not realizable by a strongly equivalent layout.
\end{proof}

\begin{figure}[htbp]	
	\centering
	\begin{subfigure}[t]{0.45\textwidth}
		\centering
		\includegraphics[scale=1]{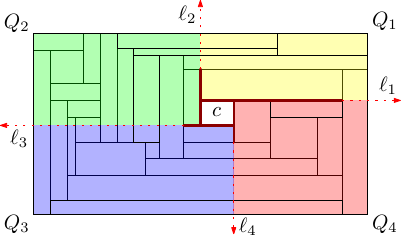}
		\caption{A nonsliceable layout, a windmill, where rays $\ell_1,\ldots , \ell_4$ define four quadrants $Q_1,\ldots , Q_4$.\label{fig:dividing:1}}
	\end{subfigure}
	\quad\quad\quad
	\begin{subfigure}[t]{0.45\textwidth}
		\centering
		\includegraphics[scale=1]{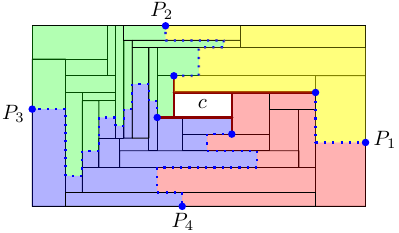}
		\caption{A weakly equivalent layout, where four paths
		define rectangular arrangements.\label{fig:dividing:2}}
	\end{subfigure}

	\caption{The rays $\ell_1,\ldots ,\ell_4$ deform into monotone paths in a weakly equivalent layout. \label{fig:dividing}}
\end{figure}

As noted above, every nonsliceable layout contains a windmill~\cite{AckermanBP06}. In the remainder of this section we strengthen this claim, and show that every nonsliceable layout contains a windmill together with pairwise disjoint $x$- or $y$-monotone paths to some points in the interiors of distinct sides of the bounding box.
In a nutshell, our proof goes as follows:
Consider an arbitrary windmill in a nonslicable layout \layout. We subdivide the exterior of the windmill into four \emph{quadrants}, by extending the arms of the windmill into rays $\ell_1,\ldots, \ell_4$ to the bounding box; see Fig.~\ref{fig:dividing:1}. Each rectangular face of \layout lies in a quadrant or in the union of two consecutive quadrants. We assign aspect ratios to each rectangular face based on which quadrant(s) it lies in. If these aspect ratios can be realized by a layout \layoutp that is weakly equivalent to \layout, then the rays $\ell_1,\ldots, \ell_4$ will be ``deformed'' into $x$- or $y$-monotone paths that subdivide \layoutp into a sublayout in the center of the windmill and four rectangular arrangements, each incident to a unique corner of the bounding box; as in Fig.~\ref{fig:dividing:2}. We assign the aspect ratios for rectangular faces in \layoutp so that these arrangements can play the same role as the rectangles $r_1,\ldots , r_4$ in a windmill layout in Figs.~\ref{fig:windmill1}--\ref{fig:windmill2}. We continue with the details.

We clarify what we mean by a ``deformation'' of a (horizontal) ray $\ell$; see Fig.~\ref{fig:dividing}.
\begin{lemma} \label{lem:deform}
Let a ray $\ell$ be the extension of a horizontal segment in a layout \layout such that $\ell$ does not contain any other segment and it intersects the rectangular faces $r_1,\ldots ,r_k$ in this order. Suppose that \layoutp is weakly equivalent to \layout, and the corresponding faces $r_1',\ldots , r_k'$ of \layoutp are sliced by horizontal segments $s_1,\ldots , s_k$, respectively. Then there exists an $x$-monotone path comprised of horizontal edges $s_1,\ldots , s_k$, and vertical edges along vertical segments of the layout \layoutp.
\end{lemma}
\begin{proof}
Assume, without loss of generality, that the ray $\ell$ points to the right. Since $\ell$ does not contain any other segment and it intersects the rectangular faces $r_1,\ldots ,r_k$ in this order, then $r_i$ and $r_{i+1}$ are on opposite sides of a vertical segment for $i=1,\ldots ,k-1$. As \layoutp is weakly equivalent to \layout, then $r_i'$ and $r_{i+1}'$ are on opposite sides of a vertical segment for $i=1,\ldots ,k-1$; see~\cite{Felsner14}. In particular, the right endpoint of $s_i$ and the left endpoint of $s_{i+1}$ are on the same vertical segment in \layoutp, for all $i=1,\ldots k-1$. These vertical segments, together with $s_1,\ldots , s_k$ form an $x$-monotone path, as required.
\end{proof}

The next lemma allows us to estimate the aspect ratio of the bounding box of a rectangular arrangement in terms of the aspect ratios of individual rectangles.

\begin{lemma} \label{lem:nonsliceable-aspect-ratios}
If every rectangle in a rectangular arrangement has aspect ratio $\alpha m$, where $\alpha>0$ and $m$ is the number of rectangles in the arrangement, then the aspect ratio of the bounding box of the arrangement is at least $\alpha$ and at most $\alpha m^2$.
\end{lemma}
\begin{proof}
Consider a rectangular arrangement $A$ with $m$ rectangles and a bounding box $R$.
Let $w$ be the maximum width of a rectangle in $A$. This implies $\wdth(R)\leq mw$.
Each rectangle of width $w$ in the arrangement has height $\alpha mw$, and so $\hght(R)\geq \alpha mw$. The aspect ratio of $R$ is $\hght(R)/\wdth(R)\geq (\alpha mw)/(mw)=\alpha$.

Similarly, let $h$ be the maximum height of a rectangle in $A$. Then $\hght(R)\leq mh$. Any rectangle of height $h$ in the arrangement $A$ has width $\frac{h}{\alpha m}$, and so $\wdth(R)\geq \frac{h}{\alpha m}$. The aspect ratio of $R$ is $\hght(R)/\wdth(R) \leq mh/(\frac{h}{\alpha m}) = \alpha m^2$, as claimed.
\end{proof}

We can now complete the characterization of aspect ratio universal layouts.

\begin{lemma}\label{lem:333}
If a layout \layout is not sliceable, it is not weakly ARU.
\end{lemma}
\begin{proof}
We proceed by induction on the number $n$ of rectangles in \layout.
In the basis step, we assume that \layout is irreducible.
In particular, \layout contains no slices, as any slice would create two nontrivial sublayouts. Every nonsliceable layout contains a windmill. Consider an arbitrary windmill in \layout, assume without loss of generality that it is clockwise and let $c$ be its central rectangle (see Fig.~\ref{fig:dividing:1}). Note that $c$ is a rectangular face in \layout, since \layout is irreducible. Denote by $R$ the bounding box of \layout. By extending the arms of the windmill into rays, $\ell_1,\ldots ,\ell_4$, we subdivide $R\setminus c$ into four \emph{quadrants}, denoted by $Q_1,\ldots , Q_4$ in counterclockwise order starting with the top-right quadrant.

Note that the interior of any rectangular face of \layout intersects at most two rays. Indeed, any two points on two different rays, say $p_i\in \ell_i$ and $p_j\in \ell_j$, span a closed axis-parallel rectangle $a(p_i,p_j)$ that intersects (the closure of) the central rectangle $c$. Now suppose that both $p_i$ and $p_j$ are in the interior of some rectangular face $r\subset R\setminus c$ of \layout, then $a(p_i,p_j)$ also lies in the interior of  $r$, and so $c$ would intersect the interior of a rectangular face $r\subset R\setminus c$, which is a contradiction. It follows that every face of \layout in $R\setminus c$ lies either in one quadrant or in the union of two consecutive quadrants.

We define an aspect ratio assignment $\alpha$ as follows: Let $\alpha(c)=1$.
If $r\subseteq Q_1$ or $r\subseteq Q_3$, let $\alpha(r)=6n$; and if  $r\subseteq Q_2$ or $r\subseteq Q_4$, let $\alpha(r)=(6n^2)^{-1}$. For a rectangle $r$ split by a ray, we set $\alpha(r)=6n+(6n^2)^{-1}$ if $r$ is split by the horizontal ray $\ell_1$ or $\ell_3$; and $\alpha(r)=((6n)^{-1}+(6n^2))^{-1}$ if split by the vertical ray $\ell_2$ or $\ell_4$.

Suppose that some layout \layoutp weakly equivalent to \layout realizes $\alpha$.
Split every rectangle of aspect ratio $6n+(6n^2)^{-1}$ in \layoutp horizontally into two rectangles of aspect ratios $6n$ and $(6n^2)^{-1}$. Similarly, split every rectangle of aspect ratio $((6n)^{-1}+(6n^2))^{-1}$ vertically into two rectangles of aspect ratios $6n$ and $(6n^2)^{-1}$; see Fig.~\ref{fig:dividing:2}. By Lemma~\ref{lem:deform}, there are four $x$- or $y$-monotone paths $P_1,\ldots ,P_4$ from the four arms of the windwill to four distinct sides of the bounding box that pass through the splitting segments. The paths $P_1,\ldots ,P_4$ subdivide the exterior of the windmill into four rectangular arrangements that we denote by $A_1,\ldots ,A_4$, and that each contain a unique corner of the bounding box. By construction, every rectangle in $A_1$ and $A_3$ has aspect ratio $6n$, and every rectangle in $A_2$ and $A_4$ has aspect ratio $(6n^2)^{-1}$.

Let $R_1,\ldots ,R_4$ be the bounding boxes of $A_1,\ldots , A_4$, respectively.
By Lemma~\ref{lem:nonsliceable-aspect-ratios}, both $R_1$ and $R_3$ have aspect ratios at least $6$, and both $R_2$ and $R_4$ have aspect ratios at most $\frac16$.
By construction, the rectangular arrangements $A_1,\ldots ,A_4$ each contain an arm of the windmill.  This implies that $\wdth(c)< \min\{\wdth(R_1), \wdth(R_3)\}$ and $\hght(c)<\min\{ \hght(R_2),\hght(R_4)\}$.
Consider the rectangular arrangement comprised of $A_1$, $c$, and $A_3$. It contains two opposite corners of $R$, and so its bounding box is $R$. Furthermore,
$\hght(R)\geq  \max\{\hght(R_1),\hght(R_3)\}$, and
\begin{align*}
\wdth(R)
&\leq \wdth(R_1)+\wdth(c)+\wdth(R_3)\\
& < 3\, \max\{\wdth(R_1), \wdth(R_3)\}\\
&\leq 3\, \max\left\{\frac{\hght(R_1)}{6}, \frac{\hght(R_3)}{6}\right\}
=  \frac{\max\{ \hght(R_1),\hght(R_3)\}}{2},
\end{align*}
and so the aspect ratio of $R$ is at least $2$.
Similarly, the bounding box of the rectangular arrangement comprised of $A_2$, $c$, and $A_3$ is also $R$, and an analogous argument implies that its aspect ratio must be at most $\frac12$.
We have shown that the aspect ratio of $R$ is at least 2 and at most $\frac12$, a contradiction. Thus the aspect ratio assignment $\alpha$ is not realizable, and so \layout is not weakly aspect ratio universal.

\paragraph{Induction step.} Assume that \layout has a nontrivial sublayout.
We distinguish between two cases:

\noindent\textbf{Case 1: \layout has a nontrivial irreducible sublayout.}
Let \layoutstar be a minimum nontrivial irreducible sublayout sublayout of \layout.
By induction, \layoutstar admits an aspect ratio assignment $\alpha^*$ that is not weakly realizable. Augment $\alpha^*$ arbitrarily to an aspect ratio assignment $\alpha$ for \layout.
A weak realization of $\alpha$ on \layout would include a weak realization of $\alpha^*$ on \layoutstar, which is impossible, so \layout is not weakly ARU.

\noindent\textbf{Case 2: all nontrivial sublayouts of \layout are sliceable.}
Then we can replace each maximal sublayout \layoutstar of \layout with a rectangle $r^*$ to obtain an irreducible layout \layoutp. By induction, there is an aspect ratio assignment $\alpha'$ for \layoutp that is not realizable. By Lemma~\ref{lem:sliceable-uniqueness} there is a suitable aspect ratio assignment $\alpha^*$ to each sliceable sublayout \layoutstar such that the bounding box of its unique weak realization has the required aspect ratio $\alpha'(r^*)$.
Thus the combination of $\alpha'$ and the assignment $\alpha^*$ yields an aspect ratio assignment for \layout that is not weakly realizable.
\end{proof}

We are now ready to prove Theorems~\ref{thm:weak} and~\ref{thm:equivalence}. We restate both theorems for clarity.

\weaktheorem*
\begin{proof}
Let \layout be a generic layout. If \layout is sliceable, then it is weakly ARU by Corollary~\ref{cor:x12}. Otherwise, it is not weakly ARU by  Lemma~\ref{lem:333}.
\end{proof}
\strongtheorem*
\begin{proof}

Properties (ii) and (iii) are equivalent by Lemma~\ref{lem:flip}.
Property (ii) implies (i) by Corollary~\ref{cor:111}. For the converse, we prove the contrapositive. Let \layout be a layout that is not one-sided or not sliceable. If \layout is not sliceable, then it is not weakly ARU by Lemma~\ref{lem:333}, hence not strongly ARU, either. Otherwise \layout is sliceable but not one-sided, and then \layout is not strongly ARU by Lemma~\ref{lem:222}.
\end{proof}

\subsection{Unique Transversal Structure}
\label{ssec:transversal}

%\paragraph{Schramm's square tiling theorem.}
Subdividing a square into squares has fascinated humanity for ages~\cite{BSST40,HenleH08,Tutte58}. For example, a \emph{perfect square tiling} is a tiling with squares with distinct integer side lengths.
Schramm~\cite{Schramm-Squares}  (see also~\cite[Chap.~6]{Lovasz}) proved that every near-triangulation with an outer 4-cycle is the extended dual graph of a (possibly degenerate or nongeneric) subdivision of a rectangle into squares. The result generalizes to rectangular faces of arbitrary aspect ratios (rather than squares):

\begin{theorem}\label{thm:Schramm} (Schramm~\cite[Thm.~8.1]{Schramm-Squares})
Let $T=(V,E)$ be a near-triangulation with an outer 4-cycle, and $\alpha\colon V^*\rightarrow \mathbb{R}^+$ a function on the set $V^*$ of the inner vertices of $T$. Then there exists a unique (but possibly degenerate or nongeneric) layout \layout such that $G^*(\LL)=T$, and for every $v\in V^*$, the aspect ratio of the rectangle corresponding to $v$ is $\alpha(v)$.
\end{theorem}
The caveat in Schramm's result is that all rectangles in the interior of every separating 3-cycle must degenerate to a point, and rectangles in the interior of some of the separating 4-cycles may also degenerate to a point. We only use the \emph{uniqueness} claim under the assumption that a nondegenerate and generic realization exists for a given aspect ratio assignment.

\begin{lemma}\label{lem:gd21}
If a layout \layout is strongly ARU,
then its extended dual graph $G^*(\LL)$ admits a unique transversal structure.
\end{lemma}
\begin{proof}
Consider the extended dual graph $T=G^*(\LL)$ of a strongly ARU layout \layout. As noted above, $T$ is a 4-connected inner triangulation of a 4-cycle. If $T$ admits two different transversal structures, then there are two strongly nonequivalent layouts, \layout and \layoutp, such that $T=G^*(\LL)=G^*(\LL')$, which in turn yield two aspect ratio assignments, $\alpha$ and $\alpha'$, on the inner vertices of $T$. By Theorem~\ref{thm:Schramm}, the (nondegenerate) layouts \layout and \layoutp that realize $\alpha$ and $\alpha'$ are unique. Consequently, neither of them can be strongly aspect ratio universal.
\end{proof}

Lemma~\ref{lem:gd21} readily shows that Theorem~\ref{thm:equivalence}(i) implies Theorem~\ref{thm:equivalence}(iii), and provides an alternative proof for the geometric arguments in Lemmata~\ref{lem:222} and~\ref{lem:333}.

\section{Recognizing Dual Graphs of Aspect Ratio Universal Layouts}
\label{sec:recognition}

%%%%%%%%%%%%%%%%%%%%%%%%%%%%%%%%%%%%%%%%%%%%%%%%%%%%%%%%%%%%%%

%The dual of a rectangular layout is called a \emph{proper graph}\todo{cite}.  It is of interest to know whether a proper graph has any dual which is aspect ratio universal (that is, sliceable and one-sided).  The following Lemmas are used in the proof of Theorem~\ref{thm:algorithm}, which presents an algorithm to check if a proper graph has an aspect ratio universal dual.\todo{Is any more introduction needed?}

In this section, we describe an algorithm that, for a given graph $G$, either finds a one-sided and sliceable (hence strongly ARU) layout \layout whose dual graph is $G$, or reports that no such layout exists (Theorem~\ref{thm:algorithm}).

Assume that we are given a near-triangulation $G$ (that is, a plane graph where every bounded face is a triangle).  A slice of a layout corresponds to an edge cut in the dual graph that contains at most two edges of the outer face. If a slice is one-sided, the edges in the edge cut form a star; and the edge cut is determined by its edges on the boundary of the outer face. A brute force algorithm would guess a slice (i.e., edge cut), and recurse on the two subproblems; it would run in exponential time. We obtain a polynomial-time algorithm using a key insight: If we already know a slice, then in each subproblem, we know two corner rectangles. Our algorithm will utilize   partial information about the corner rectangles.

\paragraph{Problem Formulation.}
The input of our recursive algorithm will be an \emph{instance} $I=(G,C,P)$, where $G$ is a near-triangulation, $C\colon V(G)\rightarrow \mathbb{N}_0$ is a \emph{corner count}, and $P$ is a set of ordered pairs $(u,v)$ of vertices on the outer face of $G$. An instance $I=(G,C,P)$ is \emph{realizable} if there exists a one-sided sliceable layout \layout such that $G$ is the dual graph of \layout, every vertex $v\in V(G)$ corresponds to a rectangle in \layout incident to \emph{at least} $C(v)$ corners of the bounding box $R$ of \layout, and every pair $(a,b)\in P$ corresponds to a pair of rectangles in \layout incident to two counterclockwise \emph{consecutive} corners of $R$.
When we have no information about corners, then $C(v)=0$ for every $v\in V$, and $P=\emptyset$.
For convenience, we also maintain the total count $C(V)=\sum_{v\in V}C(v)$, and the set $K=\{v\in V(G):C(v)>0\}$ of vertices with positive corner count.

\subsection{Structural Properties of One-Sided Sliceable Layouts}
\label{ssec:structure}
%\paragraph{Structural Properties of One-Sided Sliceable Layouts.}

Throughout this section, we use a default notation: If an instance $(G,C,P)$ is realizable by a one-sided sliceable layout \layout, then $R$ denotes the bounding box of \layout, and for every $v\in V(G)$, $r_v$ denotes the rectangular face of \layout corresponding to $v$. Note that any sublayout of a one-sided sliceable layout is also one-sided and sliceable (i.e., both properties are hereditary).

\begin{figure}[htbp]
	\centering
    \includegraphics[width=0.85\textwidth]{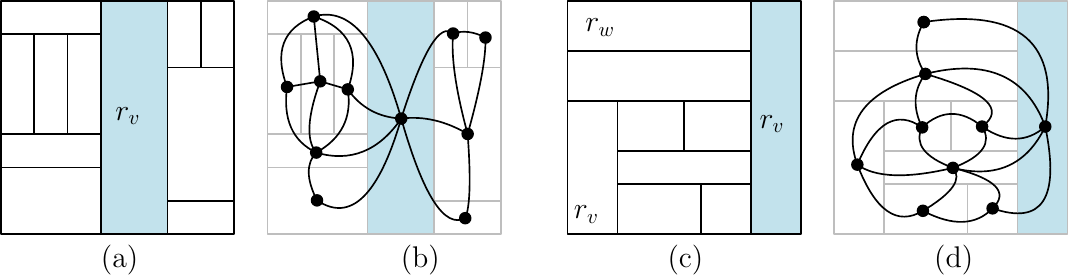}
	\caption{A one-sided sliceable layout and its dual graph $G$.
    (a--b) If $v$ is a cut vertex of $G$, then $r_v$ is bounded by two slices.
	(c--d) If there is no cut vertex in $G$, then some rectangular face $r_v$ is incident to two corners.}\label{fig:structural}
\end{figure}

\begin{lemma}
\label{lem:cut}
Assume that $(G,C,P)$ admits a realization \layout and $|V(G)|\geq 2$.
Then $G$ contains a vertex $v$ with one of the following (mutually exclusive) properties.
\begin{enumerate}
\item[(I)] Vertex $v$ is a cut vertex in $G$. Then $r_v$ is bounded by two parallel sides of $R$ and by two parallel slices; and $C(v)=0$; see Fig.~\ref{fig:structural}a--\ref{fig:structural}b).
\item[(II)] Rectangle $r_v$ is bounded by three sides of $R$ and by a slice; and $0\leq C(v)\leq 2$;
see Fig.~\ref{fig:structural}c--\ref{fig:structural}d.
\end{enumerate}
\end{lemma}
\begin{proof}
Let $v$ be a cut vertex of $G$. Then $r_v$ intersects the boundary of $R$ in at least two disjoint arcs. Since both $r_v$ and $R$ are axis-parallel rectangles and $r_v\subset R$, their boundaries can intersect in at most two disjoint arcs, which are two parallel sides of $r_v$. The other two parallel sides of $r_v$ form slices. In particular, $r_v$ is bounded by two parallel sides of $R$ and two slices, and so it is not incident to any corner of $R$. In this case, $v$ has property (I).

Assume that $G$ does not have cut vertices. Since \layout is sliceable, it is subdivided by a slice $s$ which is a segment between two opposite sides of $R$. Since \layout is one-sided, $s$ must be the side of a rectangle $r_v$ for some $v\in V(G)$. If both sides of $r_v$ parallel to $s$ are in the interior of $R$, then $r_v$ is bounded by two sides of $R$ and by two slices. Since the sublayouts of \layout on the opposite sides of these slices are disjoint, then $v$ is a cut vertex in $G$, contrarily to our assumption. Consequently, the other side of $r_v$ parallel to $s$ must be a side of $R$. Then $r_v$ is bounded by three sides of $R$ and by $s$. Clearly, $r_v$ is incident to precisely two corners of $R$, and so $v$ has property (II).
\end{proof}

Based on property (II), we say that a vertex $v$ of $G$ is a \emph{pivot} if there exists a one-sided sliceable layout \layout with $G\simeq G(\LL)$ in which $r_v$ is bounded by three sides of $R$ and a slice.
If we find a cut vertex or a pivot $v$ in $G$, then at least one side of $r_v$ is a slice, so we can remove $v$ and recurse on the connected components of $G-v$.

\paragraph{Recursive calls.}
Given an instance $I=(G,C,P)$, we define the subproblems created by removing the vertex $v$ in both cases (where $v$ is either a cut vertex or a pivot):

\textbf{(I)} For a cut vertex $v$ of $G$,
we define the operation \textsc{Split}$(G,C,P;v)$ as follows.
The graph $G-v$ must have precisely two components, $G_1$ and $G_2$. Let $(u_1,\ldots , w_1)$ and $(u_2,\ldots , w_2)$ be the sequence of neighbors of $v$ in $G_1$ and $G_2$, resp., in clockwise order. Initialize $C_1$ and $C_2$ as the restriction of $C$ to $V(G_1)$ and $V(G_2)$, respectively. Set $C_i(u_i)\leftarrow C_i(u_i)+1$ and $C_i(w_i)\leftarrow C_i(w_i)+1$ for $i\in \{1,2\}$ (if $u_i=w_i$, we increment $C_i(u_i)$ by 2). For each pair
$(a,b)\in P$, if both  $a$ and $b$ are in $V(G_i)$ for some
$i\in \{1,2\}$, then add $(a,b)$ to $P_i$. Otherwise we may assume, without loss of generality, that $a\in V(G_1)$ and $b\in V(G_2)$, and the counterclockwise path $(a,b)$ contains either $u_1,v,w_2$ or $w_1,v,u_2$. The removal of $v$ splits the path into two subpaths, that we add into $P_1$ and $P_2$, accordingly. Finally we add $(u_1,w_1)$ to $P_1$ and $(u_2,w_2)$ to $P_2$. Return the instances $(G_1,C_1,P_1)$ and $(G_2,C_2,P_2)$.

\begin{lemma}
\label{lem:split}
Let $v$ be a cut vertex of $G$. An instance $(G,C,P)$ is realizable if and only if  both instances returned by \textsc{Split}$(G,C,P;v)$ are realizable.
\end{lemma}
\begin{proof}
First assume that \layout is a realization of instance $(G,C,P)$. The removal of rectangle $r_v$ splits \layout into two one-sided and sliceable sublayouts, $\mathcal{L}_1$ and $\mathcal{L}_2$. It is easily checked that they realize $(G_1,C_1,P_1)$ and $(G_2,C_2,P_2)$, respectively.

Conversely, if both $(G_1,C_1,P_1)$ and $(G_2,C_2,P_2)$ are realizable, then they are realized by some one-sided sliceable layouts  $\mathcal{L}_1$ and $\mathcal{L}_2$, respectively. The union of a square $r_v$ and scaled copies of $\mathcal{L}_1$ and $\mathcal{L}_2$ attached to two opposite sides $r_v$ yields a one-sided sliceable layout $\mathcal{L}$ that realizes $(G,C,P)$.
\end{proof}

\noindent\textbf{(II)}
Let $v$ be a pivot of $G$. %(i.e., a vertex with property (II)).
We define the operation \textsc{Remove}$(G,C,P;v)$ as follows.
Since $v$ is not a cut vertex, then $G-v$ has precisely one component, which we denote by  $G'$. Let $(u,\ldots , w)$ be the sequence of neighbors of $v$ in $G'$ in clockwise order. Initialize $C'$ as the restriction of $C$ to $V(G')$, and then set $C'(u)\leftarrow C'(u)+1$ and $C'(w)\leftarrow C'(w)+1$.
If for any pair $(a,b)\in P$, the counterclockwise path from $a$ to $b$ in $G$ contains $v$, then return \textsc{False}. Otherwise, set $P'=P$, and add $(u,w)$ to $P'$. Return the instance $(G',C',P')$.

\begin{lemma}
\label{lem:remove}
Let $v$ be a vertex of the outer face of $G$, but not a cut vertex.
An instance $(G,C,P)$ is realizable with pivot $v$ if and only if
the instance returned by \textsc{Remove}$(G,C,P;v)$ is realizable.
\end{lemma}
\begin{proof}
Assume that $(G,C,P)$ is realized by a one-sided scliceable layout \layout, and $v$ has property (II). The removal of rectangle $r_v$ from \layout creates a one-sided sliceable sublayout $\mathcal{L}'$. It is easily checked that $\mathcal{L}'$ realizes the instance $(G',C',P')$
 returned by \textsc{Remove}$(G,C,P;v)$.

Conversely, assume that the instance $(G',C',P')$ returned by \textsc{Remove}$(G,C,P;v)$ is realized by a layout $\mathcal{L}'$.
Then let $s_{uw}$ be the maximum segment that connects the two corners of
the bounding box of $\mathcal{L}'$ incident to the $r_u$ and $r_w$, where $(u,w)\in P'\setminus P$. We can attach a single rectangular face $r_v$ to the bounding box of $\mathcal{L}'$ along segment $s_{uw}$ and obtain a layout $\mathcal{L}$ that realizes $(G,C,P)$.
\end{proof}

\paragraph{Finding a Pivot.}
In the absence of a cut vertex, any vertex incident to the outer face of $G$ might be a pivot. We use partial information on the corners to narrow down the search for a pivot.

\begin{lemma} \label{lem:two-corner-removable}
Assume that an instance $(G,C,P)$ admits a realization \layout. If $|V(G)|\geq 2$ and $C(v)\geq 2$ for some vertex $v\in V(G)$, then $v$ is a pivot.
\end{lemma}
\begin{proof}
The rectangle $r_v$ is incident to at least two corners of $R$. If $r_v$ is incident to two opposite corners of $R$, then $r_v=R$, contradicting the assumption that $G$ has two or more vertices. Hence $r_v$ is incident to two consecutive corners of $R$, and so it contains some side $s$ of $R$. The other side of $r_v$ parallel to $s$ is a maximal segment between two opposite sides of $R$, so it must be a slice.
\end{proof}

In the absence of a cut vertex, we are looking for a pivot in a 2-vertex-cut (i.e., 2-cut). Recall that a sliceable layout \layout corresponds to a BSP-tree in which every nonleaf node corresponds to a sublayout \layoutp of \layout together with a slice of \layoutp.

\begin{lemma}\label{lem:2cut}
Let \layout be a one-sided sliceable layout whose dual graph $G(\LL)$ is 2-connected and has a 2-cut $\{u,v\}$. Then there exists a one-sided sliceable layout \layoutp  such that $G(\LL')\simeq G(\LL)$ and the root slice separates the rectangles corresponding to $u$ and $v$.
In particular, $u$ or $v$ is a pivot.

Furthermore, if there are two rectangular faces in \layout that are each incident to exactly one  corner of \layout, then the corresponding rectangular faces in \layoutp are also incident to some corners in \layoutp, or there exists a one-sided sliceable layout $\LL''$ with $G(\LL'')\simeq G(\LL)$ in which one of these rectangular faces is a pivot.
\end{lemma}
\begin{proof}
Let $R$ be the bounding box of \layout; let $r_u$ and $r_v$ denote the rectangular faces corresponding to $u$ and $v$, respectively. Since $G$ is a near-triangulation, then $u$ and $v$ are adjacent in $V(G)$.
Let $s_0$ be the maximal segment that separates $r_u$ and $r_v$.
If $s_0$ connects two opposite sides of $R$, the proof is complete with $\LL'=\LL$, so we may assume otherwise. Because \layout is one-sided, whenever two rectangular faces are in contact, a side of one rectangle fully contains a side of the other. We distinguish between two cases:

\textbf{Case 1: $r_u$ and $r_v$ are in contact with opposite sides of $R$.}
We may assume, without loss of generality, that $r_u$ and $r_v$ are in contact with the bottom and top side of $R$, resp., and the bottom side of $r_v$ contains the top side of $r_u$; see Fig.~\ref{fig:opposite1}.
Since \layout is one-sided, $s_0$ is a side of some rectangular face in \layout, and we may assume that $s_0$ is the bottom side of $r_v$. The left (resp., right) side of $r_v$ lies either along $\partial R$ or in a vertical segments $s_1$ (resp., $s_2$). Since $s_0$ does not reach both left and right sides of $R$ by assumption, at least one of $s_1$ and $s_2$ exists.
Assume, without loss of generality, that $s_1$ exists. Since the bottom-left corner of $r_v$ is the endpoint of $s_0$, it lies in the interior of segment $s_1$.

We claim that the bottom endpoint of $s_1$ is on the bottom side of $R$. Suppose otherwise. Then we can incrementally construct a \emph{counterclockwise winding path} $W$: it starts with a downward edge along $s_1$. Append an edge to $W$ by making a left turn and following a maximal segment to its endpoint, until the endpoint of $W$ is on the boundary of $R$ or in the interior of a previous edge of $W$. As $W$ can cross neither $r_u$ nor $r_v$, and it cannot make a right turn if it reaches $s_0$, then it cannot reach the boundary of $R$. Consequently, $W$ ends in the interior of one of its previous edges, and the maximal segments generated by its last four edges form a counterclockwise windmill, contradicting the assumption that \layout is sliceable. This completes the proof of the claim.

\begin{figure}[htbp]	
	\centering
	\begin{subfigure}[t]{0.45\textwidth}
		\centering
		\includegraphics{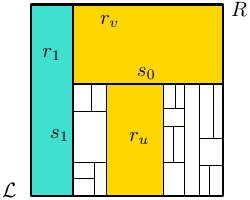}
		\caption{\label{fig:opposite1}}
	\end{subfigure}
	\hspace{1cm}
    \begin{subfigure}[t]{0.45\textwidth}
    	\centering
		\includegraphics{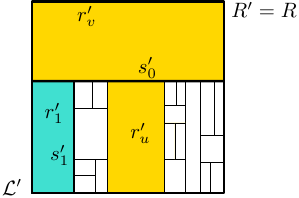}
		\caption{\label{fig:opposite2}}
	\end{subfigure}
		\caption{(a) A one-sided sliceable layout $\LL$ where $r_u$ and $r_v$ touch two opposite sides of the bounding box $R$. (b) The modified layout $\LL'$ within the same bounding box $R'=R$. \label{fig:opposite}}
\end{figure}

We can now modify \layout by extending $s_0$ and $r_v$ horizontally to the right side of $R$, and clip both $s_1$ and $r_1$ to $s_0$, as in Fig.~\ref{fig:opposite2}. This modification changes the contacts between $r_v$ and $r_1$ from vertical to horizontal, but does not change any other contacts in the layout, so it does not change the dual graph $G(\LL)$.
If segment $s_2$ exists, we can similarly extend $s_0$ to the right side of $R$ and clip $s_2$. We obtain a sliceable layout \layoutp with $G(\LL')\simeq G(\LL)$ in which the root slice $s_0'$ separates $r_u'$ and $r_v'$, and rectangle $r_v'$ is a pivot. Furthermore, every rectangle incident to a corner in \layout remains incident to some corner in \layoutp.

\textbf{Case 2: $r_u$ and $r_v$ are not in contact with opposite sides of $R$.} Then $r_u$ and $r_v$ are each in contact with a single side of $R$, and these sides are adjacent. We may assume, without loss of generality, that $r_u$ is in contact with the bottom side of $R$, $r_v$ is in contact with the left side of $R$, and the bottom side of $r_u$ contains the top side of $r_v$; refer to Fig.~\ref{fig:rk2}. Because \layout is one-sided and $s_0$ is not the root slice, then $s_0$ equals the bottom side of $r_v$.

\begin{figure}[htbp]	
	\centering
	\begin{subfigure}[t]{0.45\textwidth}
		\includegraphics{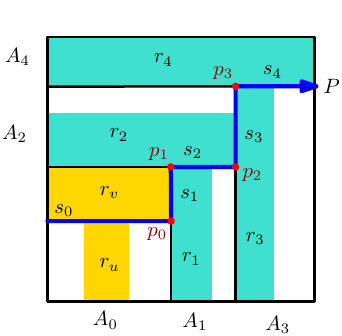}
		\caption{$\{u,v\}$ is a 2-cut of $G(\LL)$. \label{fig:rk2}}
	\end{subfigure}
	\hspace{1cm}
    \begin{subfigure}[t]{0.45\textwidth}
		\includegraphics{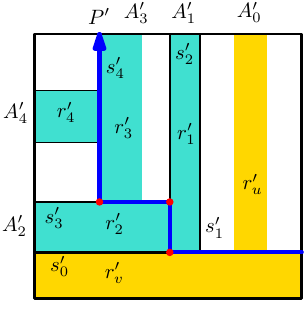}
		\caption{Layout \layoutp, where $r'_v$ is a pivot. \label{fig:rkprime2}}
	\end{subfigure}
	\caption{The construction of \layoutp from \layout.}
\end{figure}

We incrementally construct an $x$- and $y$-monotone increasing directed path $P$ (i.e., a staircase) starting with edge $s_0$, directed to its right endpoint $p_0$. Initially, let $P=\{s_0\}$ and $i:=0$. While $p_i$ is not in the top or right side of $R$, let $p_{i+1}$ be the top or right endpoint of segment $s_i$, append the edge $p_ip_{i+1}$ to $P$, and let $s_{i+1}$ be the segment orthogonal to $s_i$ that contains $p_{i+1}$. Since the path $P$ is $x$- and $y$-monotonically increasing, it does not revisit any segment. Thus the recursion terminates, and $P$ reaches the top or right side of $R$.

%Denote by $s_0,s_1,\ldots, s_k$, for some $k\geq 1$, the maximal segments of \layout that contain the edges of the directed path $P$.
Let $s_{-1}$ denote the right side of $r_u$. We claim that for $i=1,\ldots , k$, if $s_i$ is vertical, its bottom endpoint is on the bottom side of $R$, and if $s_i$ is horizontal, its left endpoint is on the left side of $R$.
The claim clearly holds for $i=0$, since $s_0$ is the bottom side of $r_v$, which is in contact with the left side of $R$. Suppose for contradiction that the claim holds for $s_{i-1}$ but not for $s_i$. Then the clockwise or counterclockwise winding path starting with $s_i$ would create a windmill (as it can cross neither $s_{i-1}$ nor $s_{i-2}$, contradicting the assumption that \layout is sliceable.

The segments $s_0,s_1,\ldots , s_k$ jointly form a one-sided sliceable layout, that is, they subdivide $R$ into $k+2$ rectangular regions, each of which contains a sublayout of \layout. One of these regions is $r_v$. Label the remaining $k$ regions by $A_0,A_1,\ldots ,A_k$ in the order in which they occur along $P$; see Fig.~\ref{fig:rk2}. In particular, we have $r_u\subset A_0$. For $i=1,\ldots , k$,  region $A_i$ is bounded by $\partial R$ and segments $s_i$, $s_{i+1}$, and $s_{i+2}$ (if they exist); and $A_k$ is adjacent to the top-right corner of $R$. Because \layout is one-sided, segment $s_i$ is a side of a rectanglular face that we denote by $r_i$, for $i=1,\ldots , k$; and $r_i\subseteq A_i$ as the opposite side of $s_i$ is subdivided by segment $s_{i-1}$.

Furthermore, we claim that $A_k=r_k$. Indeed, $A_k$ is bounded by segment $s_k$ and three sides of $B$. If $A_k\neq r_k$, then $r_k$ separates the subarrangement $A\setminus r_k$ from $r$. This means that $v_{r_k}$ would be a cut vertex in $G(\LL)$, contradicting the assumption that $G(\LL)$ is 2-connected.

We recursively construct a one-sided sliceable \layoutp by placing rectangles and subarrangement corresponding to those in \layout such that $G(\LL)\simeq G(\LL')$; refer to Fig.~\ref{fig:rkprime2}. Let $R'$ be the bounding box of \layoutp. First subdivide $R'$
by a horizontal segment $s'_0$; and let $r_v'$ be the rectangle below $s_0'$. This ensures that $s_0'$ is the root slice and $r_v'$ is a pivot. Subdivide the region  above $s_0'$ by a vertical segment $s_1'$ into two rectangular regions. Denote the right region by $A_0'$,
and subdivide the left region as follows: For $i=2,\ldots ,k$, recursively subdivide the rectangle incident to the top-left corner of $R'$ by a segment $s_i'$ orthogonal to $s_i$.

Segments $s_0',\ldots , s_k'$ jointly subdivide $R'$ into $k+2$ rectangular regions: $r_v'$ and $A_0', A_1',\ldots , A_k'$ in the order in which they are created, where $A_k'$ is incident to the top-left corner of $R'$; and all other regions are in contact with either the left or the top side of $R'$. We insert a sublayout in each region $A_i'$.
First insert a $180^\circ$-rotated affine copy of $A_0$ into $A_0'$.
For $i=1,\ldots ,k-1$, insert $r_i'$ into $A_i'$ such that its top or left side is $s_{i+1}'$; and if $A_{i-2}\setminus r_{i-2}$ is nonempty, insert an affine copy of the sublayout $A_{i-1}\setminus r_{i-1}$ into $A_k'$.
Finally, for $i=k$, subdivide $A_k'$ into three rectangles by slices orthogonal to $s_k'$:
If $A_{k-2}\setminus r_{k-2}$ or $A_{k-1}\setminus r_{k-1}$ is nonempty, insert an affine copy in the first and third rectangle in $A_k'$; and fill all remaining space by $r_k'$. This completes the construction of layout \layoutp (see Fig.~\ref{fig:rkprime2}).
By construction, we have $G(\LL')\simeq G(\LL)$.

It remains to track the rectangles incident to the corners of \layout and \layoutp.
In the original layout \layout, the rectangular face $r_k$ is incident to two corners of $R$.
Assume, without loss of generality, that $r_k$ is incident to the two top corners of $R$ (as in Fig.~\ref{fig:rk2}), and two distinct rectangular faces $r_{\mathrm{left}}\subset A_0$ and $r_{\mathrm{right}}\subset A_{k-1}$ are incident to the bottom-left and bottom-right corners of $R$, respectively.
The sublayout $A_0$ was inserted into $A_0'$ after a $180^\circ$ rotation, and so $r_{\mathrm{left}}'\subset A_0'$ is incident to the top-right corner in \layoutp.
If $r_{\mathrm{right}}\subset A_{k-1}\setminus r_{k-1}$, then
$A_{k-1}\setminus r_{k-1}$ is nonempty and it was inserted into the top third of
$A_k'$ after a $180^\circ$ rotation, and so $r_{\mathrm{right}}'$ is incident to the top-left corner in \layoutp.
Otherwise $A_{k-1}=r_{k-1}$, and then $r_{\mathrm{right}}=r_{k-1}$.
In this case  $r_k'$ is incident to the top-left corner in \layoutp. However,
we can modify \layout by extending $r_{k-1}$ and $s_{k-1}$ to the top side of $R$, and
obtain a one-sided sliceable layout $\LL''$ with $G(\LL')\simeq G(\LL)$ in which $r_{\mathrm{right}}''=r_{k-1}''$ is a pivot. This completes the proof in Case~2.
\end{proof}

\begin{lemma}
\label{lem:two-distinct-corners}
Assume that an instance $(G,C,P)$ is realizable; $G$ is 2-connected; $|V(G)|\geq 4$; there exist two adjacent vertices, $u$ and $v$, such that $C(u)=C(v)=1$, and $C(w)=0$ for all other vertices. %; and $P=\{(u,v)\}$.
Then $u$ or $v$ is a pivot; or else
$G$ has a 2-cut and a vertex of an arbitrary 2-cut is a pivot.
\end{lemma}
\begin{proof}
Let \layout be a one-sided sliceable layout that realizes $(G, C, P)$. If $r_u$ and $r_v$ contain two opposite corners of \layout, then the maximal line segment that separates them is a slice of \layout, and so $r_u$ or $r_v$ must also contain another corner, and thus be a pivot.

We may assume, then, that $r_u$ and $r_v$ contain adjacent corners of \layout, which we may assume, without loss of generality, to be the top-left and bottom-left corners, respectively.  If either spans the width of \layout and contains another corner, then it corresponds to a pivot and the proof is complete.

\begin{figure}[htbp]	
	\begin{subfigure}[t]{0.2\textwidth}
		\centering
		\includegraphics{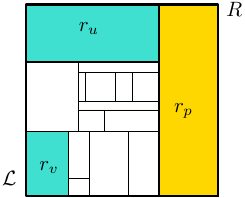}
		\caption{\label{fig:ccone1}}
	\end{subfigure}
	\hspace{5mm}
    \begin{subfigure}[t]{0.2\textwidth}
    	\centering
		\includegraphics{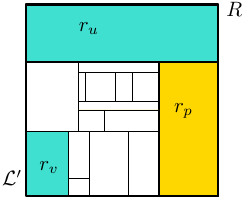}
		\caption{\label{fig:ccone2}}
	\end{subfigure}
	\hspace{5mm}
	\begin{subfigure}[t]{0.2\textwidth}
		\centering
		\includegraphics{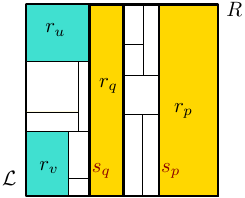}
		\caption{\label{fig:ccone3}}
	\end{subfigure}
	\hspace{5mm}
    \begin{subfigure}[t]{0.2\textwidth}
    	\centering
		\includegraphics{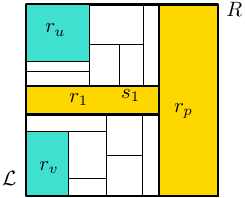}
		\caption{\label{fig:ccone4}}
	\end{subfigure}
		\caption{(a) Rectangles $r_u$, $r_u$, and $r_p$ jointly contain all four corners of $R$.
		(b) Modifying the boundary between $r_u$ and $r_p$.
		(c) If both $s_p$ and $s_q$ are slices, then $q$ would be a cut vertex.
		(d) The topmost segment $s_1$ between the left edge of $R$ and the left edge of $r_p$ is a side of $r_1$. \label{fig:ccone}}
\end{figure}

Assume that $r_u$ and $r_v$ each contain only one corner of \layout; see Fig.~\ref{fig:ccone}. There is some rectangular face $r_p$ that contains the top-right and bottom-right corners of \layout, or else there would be no pivot, contradicting Lemma~\ref{lem:cut}. If $r_u$ and $r_p$ are adjacent, then we can modify \layout by removing the right side of $r_u$ and extending the bottom side of $r_u$ to the right side of $R$; refer to Figs.~\ref{fig:ccone1}--\ref{fig:ccone2}. This yields a layout $\LL'$ realizing $(G,C,P)$ in which $r_u$ contains two corners, and thus $u$ is a pivot. The same argument can be made for $v$ being a pivot if $r_v$ and $r_p$ are adjacent.

In the remainder of the proof, we assume that neither $r_u$ nor $r_v$ is in contact with $r_p$.
The left side of $r_p$ is a slice as it connects the top and bottom sides of $R$. Note that $R$ cannot have any other slice, as it would be a side of some rectangle $r_q$, which does not contain any corner (cf.\ Fig.~\ref{fig:ccone3}); and so $q$ would be a cut vertex, contrary to out assumption that $G$ is 2-connected. This implies, using the fact that \layout is sliceable, that \layout has at least one horizontal segment from the left side of $R$ to the left side of $r_p$. Let $s_1$ be the topmost such segment; see Fig.~\ref{fig:ccone4}. As \layout is one-sided, then $s_1$ must be the side of some rectangular face $r_1$. The rectangular face $r_1$ can be neither $r_u$ nor $r_v$, since they are not in contact with $r_p$. The vertices in $G$ corresponding to $r_1$ and $r_p$ form a 2-cut, so $G$ has a 2-cut. Lemma~\ref{lem:2cut} guarantees that for any 2-cut, there is a one-sided sliceable layout \layoutp with $G(\LL')\simeq G(\LL)$ in which one of the vertices of the 2-cut is a pivot. Such a layout \layoutp realizes $(G, C, P)$ if the rectangular faces corresponding to $u$ and $v$ are each incident to some corners of \layoutp.
Since $r_u$ and $r_v$ are each incident to a single corner of \layout, then Lemma~\ref{lem:2cut} further guarantees that \layoutp is a one-sided sliceable layout that realizes $(G, C, P)$ and one of the vertices in the 2-cut is a pivot; or else $u$ or $v$ is a pivot.
\end{proof}

Recall that for an instance $I=(G,C,P)$, we defined the set $K=\{v\in V(G):C(v)>0\}$ of vertices that must contain corners in the realization. We show that if we already have 3 or 4 vertices in $K$, then $G$ has a cut vertex or a vertex in $K$ is a pivot.

\begin{lemma}\label{lem:no-cut-vertices}
Assume that $(G,C,P)$ is realizable and $|V(G)|\geq 2$.
\begin{enumerate}
\item If $|K| = 4$, then $G$ has a cut vertex.
\item If $|K| = 3$, then $G$ has a cut vertex or some vertex $v \in K$ is a pivot.
\end{enumerate}
\end{lemma}
\begin{proof}
Let \layout be a realization of $(G,C,P)$.
If $G$ has a cut vertex, the proof is complete. Assume otherwise.
Then, by Lemma~\ref{lem:cut}, there exists a rectangular face $r_u$ is incident to two corners of \layout, and in particular $u\in V(G)$ is a pivot.
As $R$ has only four corners, each of which is incident to a unique rectangular face in \layout, then at most two additional rectangular faces in \layout are incident to corners, hence  $|K|\leq 3$.

Assume that $|K| = 3$. Since $R$ has only four corners, each of which is incident to a rectangular face in \layout, the pivot $u$ is one of the three vertices in $K$.
\end{proof}

\subsection{Recognition Algorithm}
\label{ssec:alg}
We are now ready to prove the main result of this section.

\algotheorem*
%\begin{theorem}
%\label{thm:algorithm}
%We can decide in $O(n^2)$ time whether a given graph $G$ with $n$ vertices is %the dual of a one-sided sliceable layout.
%\end{theorem}
\begin{proof}
Given a graph $G$, we can decide in $O(n)$ time whether $G$ is a proper graph~\cite{Hasan0K13,Nishizeki013,RahmanNN98,RahmanNN02}. If $G$ is proper, then it is a connected plane graph in which all bounded faces are triangles. Let an initial instance be $I=(G,C,P)$, where $C(v)=0$ for all vertices $v$, and $P=\emptyset$.
We run the recursive algorithm \ALG$(G,C,P)$ below.

\begin{algorithm}[htp]\DontPrintSemicolon
\ALG$(G,C,P)$\;
 \SetKwInOut{Input}{input}
 \SetKwInOut{Output}{output}
\Input{a near triangulation $G$, a corner count $C\colon V(G)\rightarrow \mathbb{N}_0$, and a set $P$ of ordered pairs of vertices on the outer face of $G$}
\Output{whether $G$ has an aspect ratio universal dual}
    \Begin{
        \If{$|V(G)| = 1$}{\Return{True}\;}
        \ElseIf{$G$ has a vertex $v$ with $C(v) > 2$}{\Return{False}\;}
        \ElseIf{$G$ has a cut vertex $v$}{\textsc{Split}$(G,C,P;v)$ yields $(G_1,C_1,P_1)$ and $(G_2,C_2,P_2)$\;
        \Return{\ALG$(G_1,C_1,P_1)$ $\land$ \ALG$(G_2,C_2,P_2)$}}
        \ElseIf{$G$ has a vertex $v$ with $C(v)=2$}{        \Return{\ALG$($\textsc{Remove}$(G,C,P;v))$}}
        \ElseIf{$P=\{(u,v)\}$ with $C(u)=C(v)=1$ and $|K|=2$}{\ForEach{$w\in \{u,v\}$}{
                \If{\ALG$($\textsc{Remove}$(G,C,P;w))$}{\Return{True}}
                }
                \If{$G$ has a 2-cut}
                {Let $\{w_1,w_2\}$ be an arbitrary 2-cut of $G$\\
                {\ForEach{$w\in \{w_1,w_2\}$}{
                \If{\ALG$($\textsc{Remove}$(G,C,P;w))$}{\Return{True}}
            }}}}
        \ElseIf{$|K| = 3$}{\ForEach{$v \in K$}{
               \If{\ALG$($\textsc{Remove}$(G,C,P;v))$}{\Return{True}}}}
        \ElseIf{$|K| =0$}{\ForEach{vertex $v$ of the outer face of $G$}{
                \If{\ALG$($\textsc{Remove}$(G,C,P;v))$}{\Return{True}}
                }}
        \Return{False}}
\end{algorithm}

\smallskip\noindent\emph{Correctness.}
We argue that algorithm \ALG$(G,C,P)$ correctly reports whether an instance $(G,C,P)$ is realizable.

\textbf{Lines~3--4.}
A graph with only one vertex corresponds to a layout containing a single rectangle, which is clearly aspect ratio universal.

\textbf{Lines~5--6.}
A rectangle that contains three or more corners of a layout must be the only rectangle in the layout. However, the algorithm reaches this step only if there are multiple vertices in the graph (Lines 3--4), so a vertex with a corner count of 3 or more is a contradiction.

\textbf{Lines~7--9.} Lemma~\ref{lem:split} established that an instance is realizable if and only if both instances produced by a \textsc{Split} operation are realizable.

\textbf{Lines~10--26.}
In the absence of a cut vertex, we try to find a pivot. Lemma~\ref{lem:remove} established that an instance is realizable if and only if it remains realizable after removing a pivot.
\begin{itemize}
    \item[(A)] \textbf{Lines~10--11.} Lemma~\ref{lem:two-corner-removable} shows that a vertex $v$ with $C(v)=2$ is a pivot.

    \item[(B)] \textbf{Lines~12--20.}   Lemma~\ref{lem:two-distinct-corners} shows that if $C(u)=C(v)=1$ and the corner count of all other vertices is zero (hence $K=\{u,v\}$ and $|K|=2$), then $u$ or $v$ is a pivot, or at least one vertex in every 2-cut is a pivot. The algorithm tests whether $u$ or $v$ is a pivot, or any vertex of an arbitrary 2-cut is a pivot.

    \item[(D)] \textbf{Lines~21--24.}
        By Lemma~\ref{lem:no-cut-vertices}, when $|K| = 3$, a vertex in $K$ must be a cut vertex or a pivot, or else the instance is not realizable.
    \item[(E)] \textbf{Lines~25--28.}
        If we have no information about the corners and there is no cut vertex (that is, $C(v)=0$ for all $v\in V(G)$, hence $K=\emptyset$), then one of the vertices in the outer face must correspond to a pivot by Lemma~\ref{lem:cut}, or else $(G,C,P)$ is not realizable.

\end{itemize}

\textbf{Line~29.}
If we find neither a cut vertex nor a pivot, then the instance is not realizable by Lemma~\ref{lem:remove}.

\paragraph{Data Structures.}
Algorithm \ALG$(G,C,P)$ recursively removes a cut vertex or a pivot $v$ of $G$, and recurses on connected components of $G-v$. The recursive calls are supported by an in-place data structure that dynamically maintains the instance $(G,C,P)$. Specifically, we maintain $G$ in the classical \emph{doubly-link edge list} (for short \emph{DCEL}) data structure~\cite{GS85,Wei85}; see also~\cite{Kettner99}. It maintains the incidences between vertices, edges, and faces; and supports edge deletion in $O(1)$ time, hence the deletion of a vertex $v$ in $O(\deg(v))$ time.

We augment the classical DCEL data structure with three binary indicator variables for (1) vertices of the outer face, (2) cut vertices, and (3) edges corresponding to 2-cuts. The DCEL data structure already maintains the cyclic list of vertices incident to the outer face. The algorithm maintains the property that $G$ is an internally triangulated plane graph. This in turn implies that a vertex $v$ is a cut vertex if and only if the outer face appears at least twice in the cyclic order of all faces incident to $v$. Similarly, $\{u,v\}$ is a 2-cut if and only if $uv$ is an edge and the outer face appears at least twice in the cyclic order of all faces incident to $uv$. Consequently, we can identify new cut vertices (resp., edges corresponding to 2-cuts) in $O(1)$ time whenever a new vertex or edge becomes incident to the outer face. The \textsc{Split} and \textsc{Remove} operations each remove a cut vertex or a pivot, which is incident to the outer face, and so the distance of any vertex to the outer face monotonically decreases. This implies that each indicator variable changes at most once, hence all three variables can be maintained in $O(n)$ total time in any descending path of the recursion tree.

\paragraph{Runtime Analysis.}
Let $T$ be the recursion tree of the algorithm  \ALG$(G,C,P)$ for some initial instance $(G,C,P)$. The number of vertices in $G$ strictly decreases along each descending path of $T$, and so the depth of the tree is $O(n)$.

We distinguish between two types of nodes in $T$: If a step in lines~8--9 is executed, then the vertex set $V(G)$ is partitioned among the recursive subproblems; we call these \emph{partition nodes} of $T$. In the steps in (B) lines~12--20, (D) lines~21--24, and (E) lines~25--28, however, $|V(G)|-1$ vertices appear in all four, three, or $O(|V(G)|)$ recursive subproblems; we call these \emph{duplication nodes} of $T$. The algorithm performs a DFS traversal of $T$ sequentially, and so the same in-place DCEL data structure of size $O(n)$ can support the entire course of the algorithm.

We first analyze the special case that $T$ does not have duplication nodes. Then $T$ is a binary tree with $O(n)$ nodes. Overall, the total time taken by maintaining $G$ (in the DCEL data structure) and the annotation $C$ and $P$ is $O(n)$ over the course of the algorithm.

Next, we analyze the impact of duplication nodes.
We claim that steps (B), (D) and (E) are reached at most once.
Note first that the total corner count $C(V)=\sum_{v\in V}C(v)$ monotonically increases along any descending path of $T$.
In the initial instance, we have $C(V)=0$. All \textsc{Remove} and \textsc{Split} operations produce subproblems with $C(V)\geq 2$.
Every \textsc{Remove} operation removes one vertex $v$ with $C(v)\leq 2$ and increments the total corner count by two.
%Each \textsc{Split} operation removes a cut vertex (with corner count zero), and increments the corner count of each subproblem by two:

Due to the steps in (A) lines 10-11, we may assume that $C(v)\in \{0,1\}$ for all $v\in V(G)$, hence $C(V)=|K|$ in steps (B), (D), and (E).
Step (B) can be reached only when $C(V)=2$ and two distinct vertices have positive corner counts. In this case, operation \textsc{Remove}$(G,C,P;v)$ in line~14 or~19 removes at most one of these vertices, and increments the corner counts by two, yielding a subproblem with $C(V)=3$. Similarly, step (D) can be reached only when $C(V)=3$, and produces a subproblem with $C(V)=4$.
Finally, step (E) can be reached when $C(V)=0$, and \textsc{Remove}$(G,C,P;v)$ in line~27 produces a yields a subproblem with $C(V)=2$.
As steps (B), (D), and (E) produce at most four, three, and $O(|V(G)|)$ recursive subproblems, resp., the duplication nodes increase the upper bound on the runtime by a factor of $4\cdot 3\cdot n=12n$, hence it is $O(n^2)$.
\end{proof}

\section{Conclusions}
\label{sec:con}

We have shown that a layout \layout is weakly (resp., strongly) ARU if and only if \layout is sliceable (resp., one-sided and sliceable); and one can decide in $O(n^2)$ time whether a given graph $G$ on $n$ vertices is the dual graph of a one-sided sliceable layout. An immediate open problem is whether the runtime can be improved. Recall that no polynomial-time algorithm is currently known for recognizing the dual graphs of sliceable layouts~\cite{DasguptaS01,KustersS15,YeapS95} and one-sided layouts~\cite{EppsteinMSV12}.
It remains open to settle the computational complexity of these problems.

Cut vertices and 2-cuts play a crucial role in our algorithm. We can show (Proposition~\ref{pro:cuts} below) that the dual graphs of one-sided sliceable layouts have vertex cuts of size at most three. In contrast, the minimum vertex cut in the dual graphs of one-sided layouts (resp., sliceable layouts) is unbounded.
Perhaps 3-cuts can be utilized to speed up our algorithm.

\begin{proposition}\label{pro:cuts}
Let $G$ be the dual graph of a one-sided sliceable layout. If $|V(G)|\geq 4$, then $G$ contains a vertex cut of size at most 3.
\end{proposition}
\begin{proof}
Let \layout be a one-sided sliceable layout with $n\geq 4$ rectangular faces in a bounding box $R$, and with dual graph $G=G(\LL)$. For a rectangular face $r$ in \layout, let $v(r)$ denote the corresponding vertex in $G$. If $G$ is outerplanar, then either $G$ has a cut vertex, or $G$ is a triangulated $n$-cycle, hence any diagonal forms a 2-cut.
We may assume that $G$ has an interior vertex.

Consider a sequence of segments that recursively subdivide $R$ into the layout \layout; and let us focus on the first subdivision step that creates a rectangle $R_0$ in the interior of $R$. We may assume, without loss of generality, that $R_0$ is bounded by the segments $s_1,\ldots ,s_4$ in counterclockwise order, the subdivision along $s_4$ splits some rectangle into $R_0$ and $R_1$, and $s_4$ is the bottom side of $R_0$. Then the two endpoints of $s_4$ lie in the relative interiors of $s_1$ and $s_3$. Since \layout is one-sided, $s_1$ is the right side of some rectangular face $r_1$, and $s_3$ is the left side of some rectangular face $r_3$ of \layout.
Since $R_0$ is the first rectangle in the interior of $R$,
both $s_1$ and $s_2$ has an endpoint on the boundary of $R$; see Fig.~\ref{fig:3cut1}.

\begin{figure}[htbp]	
	\centering
	\begin{subfigure}[t]{0.3\textwidth}
		\centering
		\includegraphics{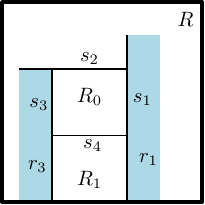}
		\caption{\label{fig:3cut1}}
	\end{subfigure}
	\quad
    \begin{subfigure}[t]{0.3\textwidth}
		\centering
		\includegraphics{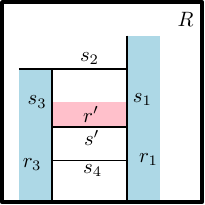}
		\caption{\label{fig:3cut2}}
	\end{subfigure}
	\quad
	\begin{subfigure}[t]{0.3\textwidth}
		\centering
		\includegraphics{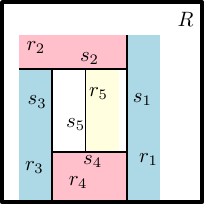}
		\caption{\label{fig:3cut3}}
	\end{subfigure}
	\caption{Schematic views of the arrangements in the proof of Proposition~\ref{pro:cuts}. \label{fig:3cuts}}
\end{figure}

We claim that at most two maximal horizontal segments in $\mathcal{L}$ intersect both $s_1$ and $s_3$ in the interior of $R$. Suppose, to the contrary, that three or more such segments (Fig.~\ref{fig:3cut2}). Let $s'$ be one of them other than the lowest or highest. Since \layout is one-sided, $s'$ is a side of some rectangular face $r'$ in \layout, which is adjacent to both $r_1$ and $r_3$, but not adjacent to the boundary of $R$, hence $\{v(r_1),v(r'),v(r_3)\}$ is a 3-cut in $G$.

In the remainder of the proof, we may assume that at most two maximal horizontal segments intersect both $s_1$ and $s_3$ in the interior of $R$. Consequently $s_2$ ($s_4$) is the highest (lowest) such segment. Since \layout is one-sided, $s_4$ is the side of some rectangular face $r_4$ of \layout. If $s_4$ is the bottom side of $r_4$, then  $\{v(r_1),v(r_2),v(r_3)\}$ is a 3-cut in $G$. We may assume that $s_4$ is the top side of a rectangle $r_4$ in \layout. Since no segment intersects both $s_1$ and $s_3$ below $s_4$, then $r_4=R_1$  (as in Fig.~\ref{fig:3cut3}).

If $R_0$ is not sliced further recursively, then
$\{v(r_1),v(R_0),v(r_3)\}$ is a 3-cut in $G$. Otherwise, let $s_5$ be the first segment that slices $R_0$. Segment $s_5$ cannot be horizontal, as it would intersect both $s_1$ and $s_3$, contrary to the assumption above. So $s_5$ is vertical (Fig.~\ref{fig:3cut3}), and it is a side of some rectangular face $r_5$ of \layout, which is adjacent to $s_2$ and $s_4$. This further implies that $s_2$ is the bottom side of a rectanglular face $r_2$ of \layout.
If $r_2$ is adjacent to the boundary of $R$, then $\{v(r_4),v(r_5),v(r_2)\}$ is a 3-cut; else $r_2$ lies in the interior of $R$ and $\{v(r_1),v(r_2),v(r_3)\}$ is a 3-cut in $G$.
\end{proof}

%\bibliographystyle{plainurl}
%\bibliography{fullARU}

\end{document}